\documentclass[submission,copyright,creativecommons]{eptcs}
\providecommand{\event}{NCMA 2024} 
\usepackage{latexsym}
\usepackage{amsfonts}
\usepackage{amsmath}
\usepackage{amssymb}
\usepackage{array}
\usepackage{epstopdf}
\usepackage{amsthm}
\newtheorem{theorem}{Theorem}
\newtheorem{corollary}{Corollary}

\usepackage{iftex}

\ifpdf
  \usepackage{underscore}         
  \usepackage[T1]{fontenc}        
\else
  \usepackage{breakurl}           
\fi

\title{Large Language Models and\\ the Extended Church-Turing Thesis\thanks
{The research of the first author was partially supported by Grant No. CK04000150 EBAVEL of the Czech Technology Agency, programme Strategy AV21 ``Philosophy and Artificial Intelligence", and the Karel \v Capek Center for Values in Science and Technology.}}
\author{Ji\v{r}\'{\i} Wiedermann
\institute{Institute of Computer Science\\ Czech Academy of Sciences\\
Prague, Czech Republic}
\email{jiri.wiedermann@cs.cas.cz}
\and
Jan van Leeuwen
\institute{Department of Information \\and Computing Sciences, \\Utrecht University, the
Netherlands}
\email{\quad J.vanLeeuwen1@uu.nl}
}
\def\titlerunning{LLMs and the Extended Church-Turing Thesis}

\def\authorrunning{Ji\v{r}\'{\i} Wiedermann and Jan van Leeuwen}

\begin{document}
\maketitle

\begin{abstract}
The Extended Church-Turing Thesis (ECTT) posits that all effective information processing, including unbounded and non-uniform interactive computations, can be described in terms of interactive Turing machines with advice. Does this assertion also apply to the abilities of contemporary large
language models (LLMs)? From a broader perspective, this question calls for an investigation of the computational power of LLMs by the classical means of computability and computational complexity theory, especially the theory of automata. Along these lines, we establish a number of fundamental results. Firstly, we argue that any fixed (non-adaptive) LLM is computationally equivalent to a, possibly very
large, deterministic finite-state transducer. This characterizes the base level of LLMs. We extend this to a key result concerning the simulation of space-bounded Turing machines by LLMs. Secondly, we
show that lineages of evolving LLMs are computationally equivalent to interactive Turing machines with advice. The latter finding confirms the validity of the ECTT for lineages of LLMs. From a computability viewpoint, it also suggests that lineages of LLMs possess super-Turing computational power. Consequently, in our computational model knowledge generation is in general a non-algorithmic process realized by lineages of LLMs.  Finally, we discuss the merits of our findings in the broader context of several related disciplines and philosophies.
\end{abstract}

\section{Introduction}

\noindent{\bf Historical context}
Back in 2001, at the occasion of entering a new millennium, a notable book entitled {\em ``Mathematics Unlimited --- 2001 and Beyond"} appeared, giving a unique overview of the state of mathematics at the end of the twentieth century and offering remarkable insights into its future development and that of its related fields. In one of the book chapters, we investigated the role of the classical Turing machine paradigm in contemporary computing  \cite{vLW:01}. Especially, we were interested in whether the {\em Church-Turing Thesis} (CTT), claiming that every algorithm can be described in terms of a Turing machine, has withstood the `test of time'. Did this thesis still apply to all modern computations as we witness them today? We characterized the latter as computations that fulfill three conditions that appeared with the advent of modern computing technologies: {\em non-uniformity of programs, interactivity} (or {\em reactivity}), and {\em infinity of operation.} Based on this observation, we brought evidence and argumentation in favor of what we called the {\em Extended Church-Turing Thesis} (ECTT), which can be seen as a strengthening of the classical Church-Turing Thesis from this newer perspective:
\begin{quote}
{\em {\bf Extended Church-Turing Thesis:} All effective information processing, including unbounded and non-uniform interactive computations, can be modeled in terms of interactive Turing machines with advice.}
\end{quote}

\noindent
Interactive Turing machines with advice (ITM/As) are variants of Turing machines that accommodate
the above-mentioned, most general conditions: non-uniformity of programs (covered by the advice mechanism), interactivity, and potentially unbounded operation \cite{vLW:01}. The
 extended thesis still perceives Turing machines as etalons of all underlying computations.

Examples of effective processes covered by the ECTT, aside from those covered by the classical CTT, include artificial computational systems like the Internet, sensory nets, ad-hoc nets, cognitive systems, relativistic computing, and natural data processing in DNA, cells, brains, and the Universe, and similar systems.

Since the formulation of the ECTT, more than 20 years have passed. Today, new pretenders to the most complex software systems have emerged: large language models (LLMs). Namely, as the pinnacles of contemporary information technology, LLMs mark a paradigm shift in the AI landscape, offering so-far unseen properties: emerging autonomy, and amazing abilities of natural language processing and understanding. Some even see LLMs as messengers of artificial general intelligence (AGI) \cite{Ag:2023}.

\medskip\noindent{\bf Motivation}
From this perspective, the following questions suggest themselves. Does the ECTT apply to the effective behavior of LLMs, i.e., can ITM/As simulate LLMs? And if so, can LLMs simulate any ITM/As, or is it beyond their power?

These questions seem to be of purely academic interest, but answering them is significant for several reasons. For one, a positive answer to the first question would confirm the thesis's validity for LLMs and the ITM/A would keep its position as a lighthouse pointing the pathways in modern computing. This in turn would confirm the fundamental position of Turing machines as a model of computation.

The second question is challenging as well. Namely, it calls for mutual simulations between both kinds of devices. Alas, it is hard to imagine more opposing models. An ITM/A is an elegant mathematical abstraction of purposeful information processing aiming at maximal computational efficiency, whereas an LLM is a pinnacle of contemporary real information processing technology.

Yet both devices share a common inspiration: the human brain and its information-processing abilities. From this perspective, both devices are complementary. On the one hand, ITM/As seem to be suitable mainly for modeling the brain body-controlling abilities by performing mechanistic, non-uniform logic-inspired processing of non-verbal information, using potentially unbounded additional memory. On the other hand, LLMs epitomize a biologically inspired approach mimicking the ability of the human brain to learn, produce, and understand natural languages (i.e., verbal information) by finite-state machines. Can one imagine a conceptually more distant pair of devices?

\medskip\noindent{\bf Contributions and results}
Along the previous lines, the paper contributes to the present state of knowledge about the computational power of LLMs, by answering the questions as implied by the ECTT.

First, we present mutual simulations between interactive finite-state transducers and fixed (trained and non-adaptive) LLMs. This characterizes the computational power of both kinds of devices by proving their computational equivalence, within the limit of the feasible context windows of the LLMs.

Second, we present a fundamental simulation of a space-restricted TM by a ``standard", sufficiently large LLM. The vocabulary and word-to-vector embedding mechanisms of the simulated LLM must be adjusted for processing the ``language" represented by the computations of the simulated TM. Nevertheless, the necessary adjustments do not intervene in the standard architecture and internal workings of LLMs. The simulation also reveals the interdependence between the space complexity of the simulated TM and the size of the word-to-vector representation used by the simulating LLM.

Third, the insights gained from the previous simulations enable the design of mutual simulations between ITM/As and {\em lineages} of evolving LLMs, i.e.\ of sequences of ever more capable LLMs. It leads to a characterization of the computational power of the latter: lineages of LLMs possess super-Turing computational power. In the recent investigations of the computational complexity of LLMs (e.g.\ \cite{Sch:23,Stro:24,Tai:22}), the results characterizing the computational power of evolving LLMs seem to be the first to deal with the general complexity-theoretic aspects of this kind of device. If we see the LLMs as knowledge generating devices, then their super-Turing abilities indicate that in this computational model, knowledge is non-computational: it cannot be generated in general by computers satisfying the classical Church-Turing Thesis.

The paper concludes with a discussion of the amazing knowledge-generating capabilities of (very) large finite state systems, their limits, and the influence of the paradigm change in AI brought forth by the technologies used by LLMs.

Closing this introductory section, a word of warning concerning the paper's presentation style is in place. We see our paper as a pioneering, pre-formal paper making a preliminary inquiry into the newly born field of LLMs from the viewpoint of the computability theory. Necessarily, it is less formal since the basic concepts related to LLMs and their understanding are still in statu nascendi. Priority is given to the investigation of the scope of the new research field, the potential of automata theory to solve the related problems, and their merits for associated disciplines and philosophies. Therefore, the paper presents only the outlines of the basic concepts, proofs and contemplations on the respective achievements.

\section{Simulations between LLMs and finite-state transducers}
When it comes to characterizing the computational power of LLMs, note that a single LLM cannot in principle simulate an arbitrary Turing machine, simply because LLMs are finite-state devices and no such device can simulate an infinite-state machine. Thus, to characterize the computational power of LLMs, we must look for weaker models than Turing machines. Finite-state transducers (FSTs) offer themselves as a natural choice. We will argue first that deterministic FSTs and fixed (non-adaptive) LLMs are computationally equivalent, within the input limits of the latter.

As far as the computational model of an LLM is concerned, we will consider a standard, or nowadays one might even say, `classical' elementary ChatGPT model as described, e.g., in \cite{Wo:23}. These models are able to learn probability distributions of words (tokens) during a training phase, and make use of it during their inference phase. Probabilistic mechanisms are not included in the LLM architecture. During inference, their behavior is influenced by their initial setting and a learned probabilistic distribution captured from the training data.

From a computability viewpoint, trained LLMs can be viewed as deterministic interactive finite-state systems producing knowledge in response to the initial prompts. They work with finite precision weights and learned fixed statistical behavior, devoid of any memory or functional augmentations. They read their inputs in a sequential manner. LLMs require substantial computational resources for their deployment, although the development of their transformer decoder-based technology leads to ever more efficient implementations \cite{Tai:22}.

A {\em finite-state transducer} is a finite-state machine with two tapes, following the terminology for Turing machines: a two-way read-only input tape and a write-only output tape. An FST reads strings of a given set on the input tape and generates strings of a related set on the output tape. An FST is allowed to make steps that do not consume an input symbol ($\epsilon$-moves), to reflect `internal processing' of state information. A general FST can be thought of as a generator of strings, a deterministic FST as a mapping  from input to output strings.

The input-to-output properties of  FSTs relate their abilities to those of LLMs. However, the models
are still very different. For example, LLMs have an intrinsic ability to learn and adapt based on vast amounts of data and using statistical methods, whereas FSTs are designed with static rule-sets or transition tables. Also, LLMs normally process inputs from `context windows' of some bounded size, to allow for the effect of various mechanisms to generate adequate output, whereas FSTs have no such limitation on their input sequences. Techniques for `long text processing' of LLMs are rapidly advancing \cite{Dong:23,Mu:24} but not yet standard.

For a basic comparison we restrict attention to the capabilities of {\em fixed} (or, static) LLMs, i.e.,
LLMs that are trained and ready but do not change or adapt during their operation. We may assume that these LLMs are deterministic and of bounded size (in bits). To enable a mutual simulation
between the two kinds of devices, we assume that FSTs and LLMs work over the same alphabet and that an LLM's parameters are represented as finite strings of bits. In order to keep the modeling discrete and finite-state we do not assume the use of analog neural nets working with real or rational weights (cf.\ \cite{Sie:12}).

\begin{theorem}\label{thm:LLM2FST} Let $\cal L$ be a fixed LLM. Then, there is a deterministic FST $\cal T$ simulating $\cal L.$
\end{theorem}

\begin{proof} (Outline.)
Let ${\cal L}$ be a fixed LLM. We argued that $\cal L$ is deterministic and of bounded size.  The Church-Turing thesis guarantees that the deterministic system ${\cal L}$ can, in principle, be simulated by a Turing machine with the same input convention. Since ${\cal L}$ is of bounded size and hence  a finite-state system, the simulating Turing machine is of constant space complexity for all inputs and thus, computationally equivalent to some deterministic FST $\cal T.$ This proves that theoretically, a sufficiently large and complex FST could simulate the behavior of $\cal L.$
\end{proof}

The following remarks are in order. The simulation in Theorem \ref{thm:LLM2FST} can also be seen
as a philosophic thought experiment, revealing strong and weak aspects of the involved devices. On the one hand, it leads to the illuminating finding that, from a computability-theoretic point of view,
fixed LLMs are essentially deterministic finite-state transducers, which are fundamental devices of automata theory.  On the other hand, from a computational complexity point of view, a fixed LLM
with millions or billions of parameters would translate into a finite-state transducer with an astronomical number of states and transitions among them. This `state explosion' makes the resulting FST computationally intractable, but efficiency has not been our current concern. Nor has it been our concern that the resulting FST
would be incredibly complex and opaque and would not offer the same level of insight into how the
LLM arrives at its outputs, compared to analyzing the LLM's internal architecture directly.

The proof of the reverse to Theorem \ref{thm:LLM2FST} is more involved since, to simulate a deterministic FST on an LLM, one can only use the standard LLM mechanisms that are quite
incompatible with an FST's computational mechanisms. We will argue that the processing of an input string by a deterministic FST can be simulated by the `answering' of the same input as a prompt by a suitable (fixed) LLM.

\begin{theorem}\label{thm:FST2LLM} Let $\cal T$ be a deterministic FST. Then, for any $n \geq 1$, there is a fixed LLM $\cal L$ simulating $\cal T$ on all words of length at most $n$.
\end{theorem}

\begin{proof} (Outline.) Given $n$, we indicate how to construct an LLM $\cal L$ as required. Note that the given FST $\cal T$ is a deterministic two-way transducer with a finite number of states, a
finite number of input and output symbols or phrases, and a finite number of transition rules. For an LLM, this makes it a relatively simple system to learn: it  can just represent the single-valued transition function of $\cal T$ with rules $(current\_state, symbol\_read)\rightarrow (head\_move, new\_state, output\_symbol),$ for every pair $(current\_state, symbol\_read).$

The transition rules of $\cal T$  can be directly embedded into the vector representations used by an LLM. The rules can be seen as the `words' of the language that is to be processed by $\cal L$, the sentences are the sequences of `words' as they are applied in deterministic order when $\cal T$ is processing an input string. The LLM can efficiently learn the behavior of $\cal T$ because the LLM's attention mechanism can effectively capture the single-valued relationship between a tuple $(current\_ state, symbol\_read)$ and the corresponding tuple $(head\_move,new\_state, output\_symbol).$

When presented with a prompt (input) of size at most $n$, $\cal L$ stores it for reference and begins the generation of the `answer' (output string) like $\cal T$ would. The attention mechanism `predicts' (here, identifies) the unique transition rule based on the current state and input symbol, $\cal L$ outputs it as the `next' word of the output, updates its state and $O(\log n)$-size positional information, and repeats, thus generating the output string word after word (until $\cal T$ would stop or some limit is exceeded). In stead of outputting the words, $\cal L$ can `decode' them and produce the output symbol that is contained in them. This effectively makes $\cal L$ into a fixed LLM that mimics the deterministic transducer within the limit of the context windows that it can handle.
\end{proof}

There are several implicit facts one may observe in the proof. First, the proof illustrates that in general learning the prediction-of-the-next-word idea is enough. The learning of the FST makes use of exactly this idea.
Second, the language of a (deterministic) FST to be learned is composed of fragments of sequences
of state transitions covering the state diagram of the transducer. Such a state diagram consists of a finite number of cycles covering this diagram, and each computation presents a concatenation of finite paths along these cycles. Hence what has to be learned is a finite number of fragments covering the state diagram between the ``crossroads" where the cycles meet.

We return to Theorem \ref{thm:FST2LLM} in Section \ref{sec:interior}.
The two theorems lead to the following conclusion, respecting the input limitations of the LLMs.

\begin{theorem}\label{thm:FST-LLM} The computational power of fixed LLMs equals that of deterministic FSTs. \end{theorem}

Theorem \ref{thm:FST-LLM} is the ground level result for our purposes. More powerful simulations can be proved for LLMs that allow for enhanced capabilities of the transformer decoders. For an overview of this recent research area, see e.g.\ \cite{Merrill:23,Stro:24}.

\section{Simulations between LLMs and interactive Turing machines with advice}\label{sec:LLMs-ITM/As}
We now turn to our main question: the position of LLMs with respect to the ECTT. How to deal with
the development towards more and more powerful LLMs in this respect? In answering it, an important role will be played by the simulation of (deterministic) Turing machines by LLMs.

\subsection{Simulating Turing machines `inside' of LLMs}\label{sec:interior}

When contemplating the simulation of Turing machines (TMs) by LLMs, a first solution that comes to
mind is the one that led A.M.\ Turing to the design of his `Turing machine'. In this solution, the LLM
at hand is augmented with an external, potentially unbounded memory that will take the role of the tape of the simulated TM, and the LLM itself will merely serve as the finite-state control of that machine. Essentially, this is the solution presented by Schuurmans \cite{Sch:23}, who showed how to operate an external read-write memory using specific prompts to simulate computations of a universal TM. In doing so, it was not necessary to modify the LLM's weights, which the author sees as a key aspect of his proposal.

In our present approach, we strive for a different exploitation of the computational potential of LLMs, without augmenting them with any external memories -- by scrutinizing the resource limits of their computational mechanisms. This is achieved by specializing an LLM to the desired `degree' in its only task, the simulation of the given TM as far as the finite-state nature of the model allows it. Accepting the fact that LLMs are finite-state devices it is clear that, if the space complexity of the simulated TM grows with the size of its inputs, we cannot hope for its simulation by an LLM on all inputs. That is, we must accept that our simulation by an LLM of a given size can only work for a TM up to a certain fixed space complexity bound for its work tape(s). This is what we call a TM simulation `inside' of an LLM.

Now a crucial question arises that has to be answered first: where can we gain the space needed for representing a TM and its computations inside an LLM? How can we exploit the existing LLM's architecture and mechanisms, without introducing new ones? To find a solution to this problem, we look for an analogy between natural language processing (NLP) by LLMs and TM computations. Can we interpret a TM computation as a language-processing task? We consider the generic case of deterministic multi-tape TM acceptors.

Before explaining the analogy, let's look at a common representation of a successful computation by the given TM. On processing a given input, the TM generates a `sequence of configurations'. The input is written on the separate read-only input tape -- which will be presented as a `prompt' to the simulating LLM. Each configuration of the TM consists of a `listing' of its current work tape configurations. Each work tape configuration is represented by the contents of that work tape and the position of the read/write head on that tape. The input symbol currently scanned by the TM's reading head on the input tape and the current state are appended to the end of the current list of work tape configurations. The `sequence of configurations' of the TM starts with the list of initial configurations for each work tape. It ends with an accepting configuration or a configuration that exceeds the allowable space limit on the work tapes as derived from the size of the LLM. (We assume that looping is prevented, e.g.\ by timing constraints.) The transition function of the simulated TM orders
the configurations in a `valid' TM computation.

To see the analogy with natural language processing, one may view the configurations as the `words' of a fictive `Turing machine language'. The syntax of these words is given by the prescription for correct configuration representations. Sequences of such words represent sentences, or their fragments, of the underlying `Turing machine language'. The semantics of the language is given by the orderings of such words following the TM transition function. As a result, the words in a sentence are ordered according to the cause-and-effect principle, because the simulated TM is assumed to be deterministic. Any semantically correct sequence of such words represents a valid fragment of TM computations, and its processing gives it its meaning.

\begin{theorem}\label{thm:interior} Let $\cal M$ be a deterministic multi-tape TM of space
complexity $S(n)$. Then, for any $n$ and $k$ such that $S(n)\le k,$ there is an LLM $\cal L$ using
word-to-vector embeddings of size $O(k)$ simulating $\cal M$ on any input of length $n.$
\end{theorem}

\begin{proof} Given $n$ and $k$ such that $S(n) \leq k$, we construct an LLM $\cal L$ that simulates $\cal M$ on inputs of size $n$. Without loss of generality, $\cal M$ may be assumed to be
always halting. (Note that halting computations cannot be longer than the number of different configurations of $\cal M$, which is bounded by $c^{\log n + S(n)}$ for some constant $c$. This can be checked by keeping a count of the number of steps. If $S(n) \geq \log n$, this can be done within the space bound by $\cal M$ itself, otherwise  it can be done by $\cal L$ itself.) $\cal L$ will be a `standard' LLM with word-to-vector embeddings of size $O(k)$ with a special attentional mechanism to be described later in this proof.

In the training phase, our initially ``empty" LLM must be trained on various valid fragments of the TM computation for inputs of size $n$, i.e.\ with configurations of size at most $S(n)$, where $S(n)\le k$. Doing so, the set of all work tape configurations up to size $k$ will become the basic `vocabulary' (set of words) of the language of our LLM. Each configuration will serve as the word-to-vector embedding of some word from the underlying `Turing machine language'. The result of the training phase is the representation of the complete transition table for tape configurations of $\cal M$ in the LLM's memory for all inputs of size $n$ with $S(n) \leq k$.

After the training phase, the simulation of $\cal M$ on any given input of length $n$ such that $S(n) \leq k$ (given to the LLM in the form of a prompt) can start. The simulation starts from $\cal M$'s initial configuration: the input word is supplied at $\cal L$'s interface and stored for reference,  and the initial tape configurations of $\cal M$ are given in the respective word-to-vector embedding, i.e., as a word in the ``Turing machine language".

If the training phase was long enough (including valid fragments for all possible transitions) and the desired transition table of tape configurations `fits' into the LLM,  the model will generate the correct
`sentence' of consecutive words that corresponds to $\cal M$'s computation on the given input, as for each configuration there is exactly one successor configuration (because $\cal M$ is deterministic). As we assumed that $\cal M$ always halts, the generated sentence will be finite. When complete, $\cal L$ can answer by outputting `accept' or `reject' depending on the final word.

If the training phase was not long enough, which may happen if $k$ is large, then configurations may arise for which the proper transition rule is missing and is yet to be `learned'. Eventually a 100\%
correctness of the simulating LLM can be achieved by tuning its attentional mechanism (cf.\ \cite{Vas:17}) to follow the transition function among successive configurations of the simulated
TM, as in the proof of Theorem \ref{thm:FST2LLM}. There is no need to track the relations between the words (configurations) across long sequences.

Thus, $\cal L$ eventually simulates $\cal M$ on all inputs within the bounds it can handle.
\end{proof}

Theorem \ref{thm:interior} can be seen as a generalization of Theorem \ref{thm:FST2LLM} although,  strictly speaking, the simulating LLM need not be `fixed' (non-adapting). The LLM is likely to be very large. As the entire transition table for $\cal M$'s work tape configurations for inputs of size $n$ must be represented in the word-to-vector embeddings of $\cal L$, the simulating LLM is likely to have a space complexity of at least order $O(c^k)$, where $c\ge 2$ is a bound on the size of the alphabets and the state set of $\cal M$ and $S(n) \leq k$.

A further remark can be made. The proof shows that an LLM $\cal L$ can be designed to generate the chain of configurations corresponding to $\cal M$'s computation on an input, regardless of what the purpose of the computation actually is. If, for example, $\cal M$ was meant to compute a more general recursive function of the input instead, then $\cal L$ could be used equally well to obtain (an encoding of) the resulting function value that is represented in the final configuration of $\cal M$. This opens the way to the use of LLMs for `computing' arbitrary {\em recursive functions,} although it is uncommon that the LLM may well have to go through many `internal' word generations before it can actually output an answer. In several recent studies, the possible extension of LLMs to allow for precisely these extended chains of inferences are explored \cite{Merrill:23,Nye:21}.

\smallskip
We now argue that Theorem \ref{thm:interior} even holds for Turing machines with advice, a very
powerful variant of the TM model that we will employ below.  A {\em Turing machine with advice}
(TM/A) is a (deterministic multi-tape) Turing machine with an oracular input facility which, when the
machine is given any input $w$, provides the TM with an extra read-only input in the form of a finite string (`advice') that depends only on the length of $w$, i.e.\ that is the same for all inputs of the given length. Advice models the possibility that TM programs can get adjusted or modified over time, especially as input sizes increase. The advice string is placed on a separate read-only advice tape. Similar to the original input, the length of the advice is not counted into the space complexity of the respective machine.

\begin{corollary}\label{cor:interior} Let $\cal M$ be a TM/A of space complexity $S(n)$. Then, for any $n$ and $k$ with $S(n)\le k,$ there is an LLM  $\cal L$ using word-to-vector  embeddings of size $O(k)$ simulating $\cal M$ on any input of length $n.$
\end{corollary}

\begin{proof}(Outline.) Referring to the proof of Theorem \ref{thm:interior}, one can clearly add $\cal M$'s advice as an extra input without altering the argument. The advice can be stored in the embeddings used by $\cal L$ at the cost of adding only a single advice symbol to each embedding. This is the symbol read by $\cal M$ from its advice tape at the time when $\cal M$ enters the configuration represented by the respective embedding. Since in each step $\cal M$ reads at most one advice symbol, all advice symbols read during the computation of $\cal M$ will fit into embeddings that are available in $\cal L$ for simulation of $\cal M.$

The simulation proceeds similar to the proof of Theorem \ref{thm:interior}. The extension of the embeddings by advice symbols will prolong the size of each embedding of the resulting LLM by one symbol.
\end{proof}

\noindent
From the proofs it is clear that the same LLM $\cal L$ will ultimately correctly simulate $\cal M$ on inputs of every length $n$ as long as $S(n) \leq k$.

\smallskip
It is important to note that the adjustments of any LLM specialized to simulating TMs with or without advice did not put the resulting LLM outside the family of standard LLMs. When compared to LLMs that process a natural language, the necessary adjustments affect just the form of the word embeddings and the working of the attention mechanism. But the main ideas of the LLM architecture, its structure and working, have remained intact. Note that the simulating LLM in the inference phase is fixed and deterministic when fully trained.

The results clearly demonstrate that no single LLM can compute every function that a TM can. {\em No LLM is Turing complete.} This is because the size of the vector embeddings of the words in the simulating LLM must go hand in hand with the space complexity of the simulated TM, and this is not possible for fixed size embeddings. In fact, it is the consequence of the fact that any LLM is a finite-state machine and a TM generally isn't, and trivially keeps LLMs within the scope of the ECTT.

On the other hand, Theorem \ref{thm:interior} and Corollary \ref{cor:interior} give evidence of the fact that by specifying more and more `advanced' TMs, even with advice, and by increasing $n$ and $k$, more and more powerful LLMs can be constructed. It suggests that LLMs can be `scaled' to match any computational challenge they are up against. This is a possibility that must be anticipated in our further investigation.

It is an open problem whether the simulations from Theorem \ref{thm:interior} and Corollary \ref{cor:interior} can be improved. For instance, does the model need to explicitly represent the full
dictionary of the necessary ``Turing machine words"? It seems to depend on the possibilities of the internal model of an LLM. Still, one thing is sure: any simulation of an infinite-state Turing machine by a finite-state machine (like an LLM) is necessarily limited by the lack of computational resources of the latter machine, and therefore is confined to initial segments of computations of the former machines. Luckily, the simulations from Theorem \ref{thm:interior} and Corollary \ref{cor:interior} are fully sufficient for our further purposes.

\subsection{Non-uniform computation and lineages of LLMs}\label{sec:ITM/AvsLLMs}
By their very definition, the LLMs are interactive computational devices. During their operation,
future prompts can react to the answers to the previous prompts. Also, by the results from the previous section, it is conceivable that LLMs are adjusted or modified over time, or even do so themselves when needed or desired. What could, ultimately, be the computational `reach' of this conception of LLMs?

\smallskip
\noindent
{\bf Lineages} We model this very general notion of an evolving LLM by a sequence ${\mathfrak L} =
 {\cal L}_1, {\cal L}_2, \cdots $ of consecutive LLMs called a {\em lineage} (after \cite{Verb:04}).  Each member of such a sequence is specialized in performing computations that require specific `technical' parameters, such as a specific size of word embeddings, a specific input sizes and so on (like the values of $n$ and $k$ in Theorem \ref{thm:interior}).  So far, this is the standard approach as known in non-uniform complexity theory, e.g.\  in the study of Boolean circuits or neural networks.

We assume that a lineage of evolving LLMs ${\mathfrak L} = {\cal L}_1, {\cal L}_2, \cdots $ can process finite but otherwise unbounded streams of inputs as follows. Processing is initiated by ${\cal L}_1$. Suppose the processing of the current stream has progressed to LLM ${\cal L}_i$, for some $i \geq 1$,  and that a trigger of some sort is generated that the lineage must `switch' to the next `evolution' ${\cal L}_{i+1}$ of the evolving LLM. (The trigger could be a technology update, reaching a memory limit, and so on.) Then the processing is continued by ${\cal L}_{i+1}$ {\em after} it is conditioned with the `knowledge' built up by ${\cal L}_i$, possibly after being `pre-trained' ahead of time on the input stream that was processed so far. ${\cal L}_{i+1}$ only produces answers and responses to the new inputs in the stream as it receives them.

We assume that for every lineage of LLMs ${\mathfrak L} = {\cal L}_1, {\cal L}_2, \cdots $, the constituent LLMs ${\cal L}_i$ are essentially pre-trained and non-adaptive (fixed). Any change or update that is not the result of `internal' computation is assumed, in principle, to lead to a next LLM in the lineage. The action of constructing and activating a next member of $\mathfrak L$ is generally called {\em model reconstruction}.

Our question about the position of LLMs with respect to the ECTT can now be concretized as follows: what is the position of lineages of evolving LLMs with respect to the ECTT?

\smallskip
\noindent
{\bf Interactive Turing Machines with Advice}
Before answering this question, we need more details about ITM/As. An {\em interactive TM} (ITM) is a (deterministic multi-tape) Turing machine that operates on a `stream'  of input symbols, supplied
at an input port. In this mode inputs are not fixed before the computation starts but new, unforeseen inputs may appear at the input port as the computation proceeds.  Inputs may depend on outputs that were produced earlier. Moreover,  the processing of a new stream may start from the working tape configuration in which the processing of the previous stream has terminated, if it was indeed finite. For a more detailed description, cf.\ \cite{vLW:01}.

Similar to TM/As, {\em interactive Turing machines with advice} (ITM/As) \cite{vLW:01} are ITMs that are extended with an advice facility.  In this model, a new advice string may be supplied and appended to the existing advice tape, every time a new input is read and input length is increased by $1$. ITM/As arguably are the most general machine model for non-uniform interactive information processing. (Cf.\ the discussion of the ECTT in Section 1.)

\subsection{Simulation of ITM/As by lineages of LLMs and Vice Versa}
Will simulations as in Theorem \ref{thm:interior} and Corollary \ref{cor:interior} work also in the extended setting of lineages and ITMs? Of course, as long as the input streams are confined to single members of a lineage and satisfy the fixed assumptions of the theorem and the corollary, the simulations will work. But what happens when the streams fail to satisfy these assumptions, e.g.\ when streams are not bounded ahead of time? We first focus on the simulation of ITM/As.

To simulate an ITM/A by a lineage of LLMs, we must solve two problems. First, the simulation must consider the fact that the space complexity of the simulated machine may grow with the size of the input, and second, the use of advice (which depends only on the input size) must be taken into account as well. (We consider the simulation on finite but unbounded inputs only, as infinite inputs are not realistic as prompts for  LLMs.)

The general idea of the simulation is to simulate computations of the given ITM/A $\cal M$ per partes by members of a lineage of LLMs,  as the individual sequences of configurations of $\cal M$ unfold, having increasing requirements on the computational resources of the simulating LLMs.  Each sequence of configurations is simulated following Corollary \ref{cor:interior} by a dedicated member of ${\cal L}\in\mathfrak L$ as long as the configurations ``fit" into the word embeddings of $\cal L$ and the advice of $\cal M$ remains unchanged.

\begin{theorem}\label{thm:family} Let $\cal M$ be an ITM/A. Then, there exists an lineage of evolving
 LLMs  $\mathfrak L$  simulating $\cal M$ on all input streams.
\end{theorem}

\begin{proof} (Outline)
Our starting point is Corollary \ref{cor:interior}. Initially, assume that  $\cal M$ has space complexity
$S(n)$ and that we have chosen a `trigger' $k_1$ such that $S(n)\le k_1$ holds for some initial values of $n$, the number of symbols in the input stream so far. Then, by Corollary \ref{cor:interior}, there exists  an LLM  ${\cal L}_1$ using word-to-vector  embeddings of size $O(k_1)$ simulating $\cal M$ on the input stream for $n$ inputs with $n=1,2, \ldots.$ However, this simulation may have to come to a halt from two reasons.

First, for some value of $n,$ it may appear for the first time that $S(n)$ exceeds $k_1.$ This means that a configuration $\rho$  of $\cal M$  has been reached that no longer fits into the word-to-vector embeddings of size $O(k_1)$. To remedy this situation, we construct a new member ${\cal L}_2 \in \mathfrak L$ with embeddings of size $k_2>k_1.$ The respective embeddings will contain all configurations of $\cal M$ of size $k_2$ which are descendants of configuration $\rho,$ augmented, of course, with all possible inputs and advice symbols as required in the proof of Corollary \ref{cor:interior}.

Second, it might happen that for some value of $n,$ it still holds that $S(n)\le k_1,$ but that $\cal M$ gets a new advice as it reaches configuration $\rho.$ As before, this calls for a model reconstruction, this time constructing ${\cal L}_2\in\mathfrak L,$  with word embeddings of  a size $k_2$ with $k_2>k_1$ for all descendants of $\rho$ of size $O(k_2)$  and a new advice string. This also handles the case when both reasons occur simultaneously.

Proceeding inductively in the same way as indicated above, an evolving lineage ${\mathfrak L} = {\cal L}_1,{\cal L}_2, \cdots$ is obtained that simulates the ITM/A on all finite but unbounded
streams.
\end{proof}

We now consider the reverse simulation, of lineages of evolving LLMs by ITM/As. It is the simulation required for the ECTT argument.

Let ${\mathfrak L}={\cal L}_1,{\cal L}_2,\cdots$ be a lineage of evolving LLMs. Considering any LLM ${\cal L}_i$ in the lineage, it is useful to distinguish between its software and data on the one hand, and the environment in which it runs on the other. Before it `evolves', we view ${\cal L}_i$ as essentially fixed, but its `environment' can provide external sources that the LLM might use during its computation e.g.\ for probabilistic purposes. We assume that this provision  is independent of the particular input that is processed but part of the `generic' operation of ${\cal L}_i$. The LLM's action is then fully determined by its program and data (and history), if a full description of this interaction of the LLM with its environment over time is given as well. This is exactly what advice does, the rest can be simulated by an interactive Turing machine.

\begin{theorem}\label{thm:advice} Let ${\mathfrak L}={\cal L}_1,{\cal L}_2,\cdots$ be a lineage  of evolving LLMs. Then, there exists an ITM/A $\cal M$  simulating $\mathfrak L$ on all input streams.
\end{theorem}

\begin{proof} (Outline.)
The result follows from the description of how lineages work, provided an ITM/A $\cal M$ can be
designed that, for every $i \geq 1$, will simulate  the $i$-th LLM of the lineage whenever this LLM's turn has come, i.e.\ when the $i$-th switching point is passed in the processing of the input stream.

For $i \geq 1$, let the $i$-th advice of $\cal M$ be defined to be $D({\cal L}_i)$, a full description of ${\cal L}_i$ (including any provision from its environment that applies). On any input stream, if $i$ inputs have been processed, ${\cal M}$ calls its advice function to get access to  $D({\cal L}_i)$ on its advice tape, enabling it to simulate ${\cal L}_i$ when its time in the simulation of the lineage has come. Thanks to the classical Church-Turing thesis this is possible, as $D({\cal L}_i)$ is an algorithmic description of a real digital `machine'. Hence, $\cal M$ computes the same transduction as ${\cal L}_i$ on its part of the input stream.
\end{proof}

Theorems \ref{thm:family} and \ref{thm:advice} can be combined into a single statement as follows.
\smallskip

\begin{theorem}\label{thm:combined} For each lineage ${\mathfrak L }$ of evolving  LLMs there is an ITM/A $\cal M$ such that $\cal M$ simulates $\mathfrak L$ on all input streams, and vice versa.
\end{theorem}

From the point of view of computational complexity theory, Theorem \ref{thm:combined} is significant because it characterizes the computational power of `evolving LLMs'. Namely, it is known that Turing machines with advice are more powerful than classical TMs, due to the effect of advice (cf.\ \cite{vLW:01}). Therefore, Theorem \ref{thm:combined} can be said to express that lineages of LLMs have  {\em `super-Turing' computational power.} By this we do not mean that such lineages can solve undecidable problems. We merely claim that such lineages cannot be simulated by `classical' ITMs (i.e., ITMs without advice). For a more comprehensive discussion of the computational power of ITM/As, we refer to~\cite{vLW:01}.

\section{Discussion of  the amazing knowledge generation ability of very large finite-state transducers}

We now review the results we obtained from a more detached viewpoint, in the broader perspective of related fields like computability, computational complexity theory, AI theory, robotics, and cognitive sciences. The common denominator in our discussion will be to point to the potential of our findings for
a better understanding of the essential qualities and limitations of the new emerging information processing technology represented by evolving LLMs. The discussion aims to bring thought-provoking
insights, provide novel perspectives to the ongoing debates of LLMs, and challenge future AI research.

\medskip\noindent{\bf Computability and complexity aspects}
In these domains, the main message has been the confirmation that the Extended Church-Turing Thesis is valid also for the latest achievement in the field of IT technology, the development of evolving LLMs. This result could be obtained, thanks to the design of a novel simulation of ``small" (resource-bounded) Turing machines entirely within LLMs as in Theorem \ref{thm:interior} or Corollary \ref{cor:interior}. Subsequently, in the end, this has led to the proof of the super-Turing computational power of these AI systems, as a consequence of Theorem~\ref{thm:combined}. Related results comprised a complete characterization of the computational power of single  LLMs in Theorem \ref{thm:FST-LLM}, and that of lineages of LLMs in Theorem \ref{thm:combined}. These results seem to be the first results dealing with the complexity of LLMs from an automata theoretic point of view.

The results are a bit paradoxical -- mankind's most complex computational devices turn out to be computationally equivalent to one of the simplest fundamental models of computation, finite-state transducers.  LLMs are, in fact, instances of highly resourceful large scalable finite-state transducers.

\medskip\noindent{\bf The illusory language-processing power of LLMs}
In Theorem \ref{thm:FST-LLM} we saw that the computational power of LLMs is on par with that of
FSTs. This raises questions concerning the natural language-processing power of LLMs.

Namely, the languages generated or accepted by FSTs are regular. How it is then possible that, in the
practice of LLMs, these devices seem to correctly recognize long sequences of natural languages which in general are known to be more complex than words in a regular language (cf.\ \cite{Pu:82})?  This
conundrum could be explained by the fact that a finite swath of a language of whatever complexity, captured in the training set, can always be seen as part of a regular language. Beyond this swath, for
sufficiently long inputs, the language behaves as a regular language. On the one hand, this explains the apparent inherent efficiency of giant LLMs in processing natural language texts of a reasonable length we see in practice. On the other hand, it also explains the limited abilities of LLMs to deal with tasks not sufficiently represented in the training set, such as simple arithmetic, logical operations like abduction, planning, etcetera.

Nevertheless, prolonging the context window (hence the input length) indefinitely will reveal the true recognition power---that of FSTs---of LLMs, which from a certain internal configuration will start
to cycle (or halt). This seems to contradict the recently appearing articles about efficient methods to scale LLMs to infinitely long inputs (cf.\ \cite{Mu:24}). The theoretical catch here is that such
methods cannot work in fixed memory spaces like classical LLMs have. They need additional space to enable the long-span attention mechanisms to work. This space may grow with the growing input size. To overcome this difficulty either the full power of ITM/As is required or that of an infinite lineages of evolving LLMs (cf.\ Theorem  \ref{thm:family}).

Another problem with viewing LLMs as FSTs is that, in FSTs, the semantics of transitions is encoded
entirely in the ``relationships'' between the states in the state diagram, not in their ``names",
because  the states can be arbitrarily renamed (this seems to be an important observation).  On the other hand, in LLMs almost all of the semantics is encaged in the information `inside' a state. The analogy between relationships among automata states and the semantic processing of language data is hard to see.  When thinking about the semantics of words of a natural language, what is of importance is the relationships among the meanings of words, not their ``names". Here may also lay the roots of the easiness with which LLMs cross the boundaries between various existing natural languages.

\medskip\noindent{\bf The problem of understanding}
In the domain of AI, the amazingly versatile abilities of LLMs put these systems into the position of the harbingers that announce a paradigm shift afflicting the entire AI ecosystem. Our results
contribute to a better apprehension of the nature of the information processing in LLMs. The key observation in this respect is the analogy between natural language processing in LLMs and general computation realization by Turing machines. While natural language processing in LLMs works with finitely many words of a natural language, within a general TM computation each configuration is seen as a word of the ``Turing machine language". Such a language has potentially an infinite number of words. Syntax and semantics of this language are described by the underlying TM ``program". It describes the relationship between the words generated by the program and, in the end, between the input and the output of the program. In this way, it explains how the program's execution transforms the input to the output. This can be seen as a correctness proof of the program, or as a formal proof of (the machine's) understanding (of what it had done). Of course, this form of understanding is different from what we, humans, understand as understanding.

It is here where the study of programming language theory can inspire, e.g., the ongoing debates on understanding by LLMs (cf. \cite{MK:22}). The analogy between natural language processing by LLMs and the processing of the TM language by a TM may shift the debate to a firm mathematical ground.

\medskip\noindent{\bf Understanding understanding}
To illustrate the difference in understanding in LLMs and TMs, let us compare the ``mechanism of understanding" in an LLM processing a natural language, and in an LLM simulating a TM according to Theorem \ref{thm:interior}. In the former case, the decision to generate the next word of the underlying natural language is based on the limited semantic context based on the vector embeddings of several meticulously chosen words, and the gigantic linguistic background knowledge stored in the form of neural nets. In the latter case, the decision to generate the next word of the Turing machine language (i.e., the next TM configuration) is based on the entire history of computation represented by the sequence of configurations entered by the machine until that time. Moreover, any TM computation makes implicit use of the designer's background knowledge that is already embedded in the design of the underlying TM program. This kind of knowledge is tailored to the intended purpose of the computation.

Which of the two decisions concerning the prolongation of both computations being compared, is based on a more profound knowledge of the situation? The winner is the TM, because its decision is based on the maximal available information it could have directly and indirectly at its disposal.

\medskip\noindent{\bf Competence without linguistic understanding?}
From an evolutionary point of view, it seems that {\em the key to the notion of understanding is understanding in non-linguistic systems.} ``Human-like understanding" adds a layer to the understanding in the AI systems of the latter type. Contemporary wisdom is that human-like understanding is based on concepts -- internal mental models of external categories, situations, events, and one's internal state and ``self" \cite{MK:22}. LLMs can build internal representations of external categories, situations, and to some extent, one's internal state mediated to the system via textual information. Neural networks are good at building such kinds of representations and excel in verbally expressing them. However, they fail to adequately represent the events and the ``self" concept. Representation of events calls for representing the sequences of situations and actions, and the LLMs lack adequate means for doing that. Speaking about the ``self" concept is difficult in the case of disembodied entities.

Except for linguistic expressions, non-linguistic embodied AI models can deal with all the concepts mentioned before. Moreover, in the form of multi-modal cyber-physical systems, equipped with memory, sensory, motor, and feedback units, they can do more since their artificial senses are grounded in the real world. Such systems can aspire to represent, and deal with, events and realization of the concept of {\em ``what is it like, for the system, to understand"}.  Memory augmenting of AI systems may help to remember, recognize, and recall important events while multi-modality in the form of complementary external and internal sensations allows to represent the last mentioned concept that is considered to be the hallmark of consciousness  \cite{Na:74}. Such systems will find themselves on the verge of artificial phenomenal experience (cf. \cite{WvL:19,WvL:21}).

Note that the external view of non-linguistic understanding we are discussing above is ``competence without comprehension", while the internal view, from the inner perspective of the system, is ``what it is like to understand". It may well be that the latter concept presents the missing link even in our understanding to human understanding.

\medskip\noindent{\bf The inner life of LLMs}
There is more to the previous comparison between LLMs and TMs. At each computational step, an embodied TM governing a cognitive cyber-physical system has complete information available not only about its own ``current state", from all its external and internal sensors, motors, and the respective feedback from those devices that the system possesses, but also about all its past states.

Note that each configuration of such a TM contains a complete representation of the machine's ``phenomenal experience" at that time (cf.\ \cite{WvL:21}). This gives the LLM simulating a TM as in Theorem \ref{thm:interior} (that by its very definition remembers all possible ``states of mind" of such a machine, one is tempted to say) an opportunity to ``time travel" backward and forwards over these states, and thus explore its past decisions or consider its possible ``futures" and adjust its behavior accordingly. What never-thought-of possibilities for classical LLMs! This observation is of interest, especially in the context of the recent announcement on a new consensus: there is {\em ``a realistic possibility" for elements of consciousness in reptiles, insects, and mollusks} \cite{Le:24}. If so, why can't it occur in the AI systems whose complexity competes with such simple creatures? Is the feeling ``how it is like to understand?" the missing element in our understanding of understanding? In this context, considerations about minimal machine consciousness are the first signs of similar general trends in AI (cf.\ \cite{WvL:19,WvL:21}).

In general, it might be possible to consider various high-level non-linguistic cognitive abilities of LLMs via formal counterparts in the TM environment. For, how else could these abilities be ascribed to LLMs without having their mirror in the language of TMs? In this way, variants of representations of a Turing machine's  computations could serve as {\em drosophila} for ideas about LLMs.

\medskip\noindent{\bf Building a bridge between rule-based and biological computation}
The paradigm shift in our apprehension of computations mentioned above  is pregnantly expressed
precisely by Theorem \ref{thm:combined}, which builds a bridge between biologically-inspired and logico-mathematical ways of information processing. Although equipollent from a computability point of view, i.e.\ expressing the same computational power by different means, the two ways do not have equivalent significance and reach. ITM/As epitomize the classical view of computations, which are seen merely as data transformation tools. The interpretation of the results for ITM/A computations is left to its user.  The view of computations through the lens of LLMs puts stress on their semantic contents---it liberates the users from the burden of data interpretation by  ``automatizing" that task through answers in a natural language. Another view of the respective computations might be that ITM/As mostly capture the processing of non-verbal information, while LLMs capture that of linguistic, verbal information. As Browning and LeCun \cite{BLC:22} remind us, {\em ``abandoning the view that all knowledge is linguistic permits us to realize how much of our knowledge is non-linguistic"}.  Nevertheless, in both cases, these are but different forms of knowledge that are produced.  Viewing computations as knowledge generators has been coined and used by us since the last decade \cite{WvL:13,WvL:15,WvL:23}.

\medskip\noindent{\bf Semantics is all we need}
Although it may seem that LLMs let us forget about Turing Tests and Chinese Room experiments, the opposite is true. These experiments focus our attention on semantics and understanding, their importance, their representation, processing and interpretation of information that results from computation. The above-mentioned tests and experiments expose the problem of expressing the semantics of computations in their syntax and flow. Perhaps they open the problem of what are the semantic resources of computation, and how are they best represented, utilized and shared. In the case of LLMs, these seem to unequivocally be the word-to-vector embeddings. For ITM/As, the machine configurations. Is there some general theory behind this? In any case, these speculations support the view of computations as knowledge generators (cf.\ \cite{WvL:13,WvL:15}). This view
puts stress on the meaning of what is computed, rather than on how a computation is performed.

\medskip\noindent{\bf Is knowledge computable?}
The answer to this question depends on how we define ``knowledge" and ``computable". There is no one, universally agreed-upon definition of what is knowledge. Within our quest of understanding computation (cf.\ \cite{vLW:01,WvL:08,WvL:13,WvL:15}, we see computations as knowledge generators. LLMs are typical examples of computations generating knowledge \cite{WvL:23}, and especially, wisdom as the agentic form of knowledge. This is a type of knowledge that can be inferred from human-produced texts, programs, pictures, videos, various multimodal sources, and the likes.

If we accept that knowledge is what is generated by computations, and that Turing machines are recognized as the general model of computation in computability theory, then from Theorem \ref{thm:advice} it follows that knowledge generation is a non-algorithmic process that cannot be performed by the classical Turing machines. Or, to make knowledge generation computable,  shouldn't one redefine the notion of computability using interactive Turing machines with advice?

\section{Conclusion}
The era of interactive non-uniform information processing at scale is here. The Extended Church-Turing Thesis formulated as a vision more than 20 years ago, has appeared to hold for LLMs, too. The information processing by evolving LLMs heralds the advent of the new understanding of computation, and especially of AI. Despite their known deficiencies, the LLMs are wonderful, exciting, and so far quite mysterious information processing devices possessing a maximal computational power like we can expect from massive classical computations. It remains to be seen  where and what the new development of LLMs-like devices will bring us in the future. Undoubtedly, the Extended Church-Turing Thesis will cover our steps in this endeavor.


\frenchspacing
\begin{thebibliography}{99}

\bibitem{Ag:2023}
Ag\"uera y Arcas, B., Norvig, P.: Artificial General Intelligence Is Already Here. {\em NO\={E}MA}, October 13, 2023, \url{https://www.noemamag.com/artificial-general-intelligence-is-already-here/}

\bibitem{BLC:22}
Browning, J., LeCun, Y.: AI and the Limits of Language. {\em NO\={E}MA}, August 23, 2022, \url{https://www.noemamag.com/ai-and-the-limits-of-language/}

\bibitem{Dong:23}
Dong, Z., Tang, T., Li, L., Zhao, W.X.: A Survey on Long Text Modeling with Transformers. {\em arXiv preprint}, arXiv:2302.14502 (2023), \url{https://doi.org/10.48550/arXiv.2302.14502}

\bibitem{Le:24}
Lenharo, M.: Do insects have an inner life? Animal consciousness needs a rethink. {\em Nature}, April 19, 2024, \url{https://www.nature.com/articles/d41586-024-01144-y}

\bibitem{Merrill:23}
Merrill, W., Sabharwal, A.: The Expressive Power of Transformers with Chain of Thought. {\em arXiv preprint}, arXiv:2310.07923 (2023), \url{https://doi.org/10.48550/arXiv.2310.07923}

\bibitem{MK:22}
Mitchell, M., Krakauer, D.C.: The debate over understanding in AI's large language models. {\em Proceedings of the National Academy of Sciences (PNAS)}, 120 (13) e2215907120, March 23, 2023, \url{https://www.pnas.org/doi/full/10.1073/pnas.2215907120}

\bibitem{Mu:24}
Munkhdalai, T., Faruqui, M., Gopal, S.: Leave No Context Behind: Efficient Infinite Context Transformers with Infini-attention. {\em arXiv preprint}, arXiv:2404.07143 (2024),
\url{https://doi.org/10.48550/arXiv.2404.07143}

\bibitem{Na:74}
Nagel, T.: What Is It Like to Be a Bat? {\em The Philosophical Review} 83:4 (1974), 435-450,
\url{https://doi.org/10.2307/2183914}

\bibitem{Nye:21}
Nye, M., {\em et al.}: Show Your Work: Scratchpads for Intermediate Computation with Language Models. {\em arXiv preprint}, arXiv:2112.00114 (2021),
\url{https://doi.org/10.48550/arXiv.2112.00114}


\bibitem{Pu:82}
Pullum, G.K., Gazdar, G.: Natural languages and context-free languages. {\em Linguistics and Philosophy} 4 (1982) 471-504, \url{https://doi.org/10.1007/BF00360802}

\bibitem{Sch:23}
Schuurmans, D.: Memory Augmented Large Language Models are Computationally Universal. {\em arXiv preprint}, arXiv:2301.04589 (2023), \url{https://doi.org/10.48550/arXiv.2301.04589}

\bibitem{Sie:12}
Siegelmann, H.T.: {\em Neural Networks and Analog Computation: Beyond the Turing Limit}. Birkh\"auser (1999), Springer Science \& Business Media, 2012

\bibitem{Stro:24}
Strobl, L., Merrill, W., Weiss, G., Chiang, D., Angluin, D.: What Formal Languages Can Transformers Express? A Survey. {\em Trans.\ Assoc.\ Comput.\ Ling.} 12 (2024) 543-561, \url{https://doi.org/10.1162/tacl_a_00663}

\bibitem{Tai:22}
Tai, Y., Dehghani, M., Bahri, D., Metzler, D.: Efficient Transformers: A Survey. {ACM Comp.\ Surv.} 55:6 (2022) 1-28, \url{https://doi.org/10.1145/3530811}

\bibitem{vLW:01}
van Leeuwen, J., Wiedermann, J.: The Turing Machine Paradigm in Contemporary Computing. In: Engquist, B., Schmid, W. (eds), {\em Mathematics Unlimited - 2001 and Beyond}. Springer, Berlin, Heidelberg (2001), pp. 1139-1155, \url{https://doi.org/10.1007/978-3-642-56478-9\_59}

\bibitem{Vas:17}
Vaswani, A., {\em et al.}: Attention is All you Need. In: Guyon, I., {\em et al.} (eds), {\em Advances in Neural Information Processing Systems} 30 (NIPS 2017),
\url{https://doi.org/10.48550/arXiv.1706.03762}

\bibitem{Verb:04}
Verbaan, P.,  van Leeuwen, J., Wiedermann, J.: Complexity of Evolving Interactive Systems. In: Karhum\"{a}ki, J., Maurer, H., P\u{a}un, G., Rozenberg, G. (eds), {\em Theory Is Forever}, Lecture Notes in Computer Science, vol 3113. Springer, Berlin, pp.\  268-281  (2004) \url{ https://doi.org/10.1007/978-3-540-27812-2_24}

\bibitem{WvL:08}
Wiedermann, J., van Leeuwen, J. (2008). How We Think of Computing Today. In: Beckmann, A., Dimitracopoulos, C., L\"owe, B. (eds) {\em Logic and Theory of Algorithms.  CiE 2008.} Lecture Notes in Computer Science, vol 5028. Springer, Berlin, Heidelberg, pp.\ 579-593 (2008) \url{https://doi.org/10.1007/978-3-540-69407-6_61}

\bibitem{WvL:13}
Wiedermann, J., van Leeuwen, J.: Rethinking computation. In: Brown, M., Erden, Y. (Eds), 6th AISB Symp.\ on Computing and Philosophy: {\em The Scandal of Computation -- What is Computation?}, Proceedings, AISB Convention 2013, University of Exeter, pp. 6-10 (2013), \url{https://gordana.se/work/PUBLICATIONS-files/2013-PROCEEDINGS-AISB.pdf}

\bibitem{WvL:15}
Wiedermann, J., van Leeuwen, J.: What is Computation: An Epistemic Approach. In: G.F. Italiano et al. (eds), {\em SOFSEM 2015: Theory and Practice of Computer Science}, Lecture Notes in Computer Science, vol. 8939, Springer, Berlin, pp. 1-13 (2015), \url{https://doi.org/10.1007/978-3-662-46078-8\_1}

\bibitem{WvL:19}
Wiedermann, J., van Leeuwen, J.: Finite State Machines with Feedback: An Architecture Supporting Minimal Machine Consciousness. In: Manea, F., {\em et al.} (eds), {\em  Computing with Foresight and Industry}: 15th Conference on Computability in Europe (CiE 2019), Proceedings, Lecture Notes in Computer Science, Vol.\ 11558, pp.\ 286-297. Springer, Cham (2019), \url{https://doi.org/10.1007/978-3-030-22996-2\_25}

\bibitem{WvL:21}
Wiedermann, J., van Leeuwen, J.: Towards Minimally Conscious Cyber-Physical Systems: A Manifesto. In: Bure\u{s}, T., {\em et al.} (eds), {\em SOFSEM 2021: Theory and Practice of Computer Science}, Lecture Notes in Computer Science, Vol.\ 12607, pp.\ 43-55. Springer, Cham (2021), \url{https://doi.org/10.1007/978-3-030-67731-2\_4}

\bibitem{WvL:23}
Wiedermann, J., van Leeuwen, J.: From Knowledge to Wisdom: The Power of Large Language Models in AI, Technical Report UU-PCS-2023-01, Dept. of Information and Computing Sciences, Utrecht University, Utrecht, The Netherlands, 2023, \url{https://webspace.science.uu.nl/~leeuw112/techreps/UU-PCS-2023-01.pdf}

\bibitem{Wo:23}
Wolfram, S.: {\em What Is ChatGPT Doing\dots and Why Does It Work?} Wolfram Media, Inc. (2023), \url{https://doi.org/10.1007/978-3-030-67731-2\_4}

\end{thebibliography}
\end{document}